\documentclass[11pt]{article}
  
\makeatletter					

\skip\footins=4ex				
\makeatother					

\usepackage{amsthm}
\usepackage{amsmath}
\usepackage{algorithm}
\usepackage{algorithmic}
\usepackage{bbm}

\newtheorem{theorem}{Theorem}[section]
\newtheorem{claim}[theorem]{Claim}
\newtheorem{definition}{Definition}[section]
\newtheorem{property}[theorem]{Property}

\newtheorem{lemma}[theorem]{Lemma}
\newtheorem{corollary}[theorem]{Corollary}

\newcommand{\opt}{\hbox{OPT}}
\newcommand{\Rad}{\hbox{Rad}}
\newcommand{\TSP}{\hbox{TSP}}
\newcommand{\length}{\hbox{length}}
\newcommand{\level}{\hbox{level}}

\begin{document}

\title{A quasi-polynomial time approximation scheme for Euclidean capacitated vehicle routing}
\author{Aparna Das \and Claire Mathieu}
\maketitle
\abstract{}
In the capacitated vehicle routing problem, introduced by Dantzig and Ramser
in 1959, we are given the locations of $n$ customers and a depot, along with a vehicle of
capacity $k$, and wish to find a minimum length collection of tours, each
starting from the depot and visiting at most $k$ customers, whose union covers
all the customers. We give a quasi-polynomial time approximation scheme for the setting where the customers and the depot are on the plane, and
distances are given by the Euclidean metric.

\section{Introduction}
In 1959, Dantzig and Ramser introduced the vehicle routing problem (VRP) and gave a linear programming based Algorithm whose {\em``calculations may be readily performed by hand or automatic digital computing machine''}\cite{dr59}. Since its introduction, VRP has been widely studied by researchers in Operations Research and Computer Science. Several books (see \cite{tvbook}, \cite{grwbook} and \cite{fbook}, among others) have been written on the VRP and it even has its own wikipedia page \cite{wikiweb}. 

\smallskip
VRP is used to describe a class of problems where the objective is to find low cost delivery routes from depots to customers using a vehicle of limited capacity. When Dantzig and Ramser first introduced VRP they stated that {\em``no practical applications have been made as yet''}\cite{dr59}. Since that time, applications of various VRP problems have been identified in numerous industries where transportation costs matter such as food and beverage distribution, and package and newspaper delivery. Toth and Vigo report on businesses that saved between 5 \% and 20\% of total costs by solving VRP problems using computerized models \cite{tvbook}.

\smallskip
\noindent
{\bf The capacitated vehicle routing problem.} We study the most basic form of the vehicle routing problem, the capacitated version (CVRP), where the input is $n+1$ points representing the locations of $n$ customers and one depot, and a vehicle of capacity $k$. The objective is to find a collection of tours, starting at the depot and visiting at most $k$ customers, whose union covers all $n$ customers, such that the sum of the lengths of the tours is minimized. CVRP is also referred to as the $k$-tours problem in the Computer Science literature \cite{arora, aktt}.  We focus on the $2d$-Euclidean version of CVRP where customers and the depot are on the Euclidean plane.

\smallskip
\noindent
{\bf Previous work. } Several results are already known about the approximability of CVRP. When the capacity of the vehicle $k$ is $2$, it can be solved in polynomial time using minimum weight matching. The metric case was shown to be APX-complete for all $3 \le k$; Asano et al. presented a reduction from H-matching for $k = O(1)$ and there is a simple reduction from the traveling salesman problem (TSP) for larger $k$ \cite{aktt96}. Constant factor approximation with performance $(1+\alpha)$ with $\alpha$ being the best constant factor approximation for TSP, were presented by Haimovich and Rinnooy Kan \cite{hr85}.  

\smallskip

The existence of a PTAS for the $2d$-Euclidean version remains an active area of research.  Extending \cite{hr85}, Asano et al. presented a PTAS for the case $k = O(\log n/\log \log n)$ \cite{aktt}. Arora's work implies a PTAS for $k = \Omega(n)$ \cite{aktt}. 

\smallskip
\noindent
{\bf Our result. } We present a quasi-polynomial time approximation scheme for the entire range of $k$. 
 \begin{theorem}\label{thm:main}(Main Theorem) Algorithm~\ref{alg:ktour} is a randomized quasi-polynomial time approximation scheme for the $2d$-Euclidean capacitated vehicle routing problem. Given  $\epsilon > 0$, it outputs a solution with expected length $(1+O(\epsilon))\opt$, in time $n^{\log^{O(1/\epsilon)}n}$. The Algorithm can be derandomized.
\end{theorem}

\smallskip
\noindent
{\bf Where previous approaches fail. }Our approximation scheme uses the divide and conquer approach that Arora used in designing a PTAS for Euclidean TSP \cite{arora}. Like Arora, we ``divide" the problem using a randomized dissection that recursively partitions the region of input points into progressively smaller squares. We search for a solution that goes back and forth between adjacent squares a limited number of times and always through a small number of predetermined sites called \emph{portals} that are placed along the boundary of squares. It is natural to attempt to extend the TSP structure theorem and show that there exists a near optimal solution that crosses the boundary of squares a small number of times, and then use dynamic programming.  Unfortunately ~\cite{arorasurvey}, {\em \begin{quote}``we seem to need a result stating that there is a near-optimum solution which enters or leaves each area a small number of times. This does not appear to be true. [...] The difficulty lies in deciding upon a small interface between adjacent squares, since a large number of tours may cross the edge between them. It seems that the interface has to specify something about each of them, which uses up too many bits.''\end{quote}}  Indeed, to combine solutions in adjacent squares it seems necessary to remember the number of points covered by tour segments, and that is is too much information to remember. 
 


\smallskip
\noindent
{\bf Overview of our approach. } To get around this problem we remember \emph{approximately} how many points are on each tour segment. Specifically we round the number of points on each tour segment down to a multiple $(1+\epsilon/\log n)$, which means there are only $O(\log k)$ possibilities for the approximate number of points on each tour segment. This enables us to deal with the difficulties described by Arora: now we have a small interface between adjacent squares, namely, for every pair of portals and every approximate number of points, we remember the number of tour segments that have this profile. The quasipolynomial running time of our dynamic program (DP) follows since the number of profiles is polylogarithmic and as there are at most $n$ tour segments of each profile.

One challenge with only remembering the approximate number of points per tour segment is that the tours found by our dynamic program (DP) could now cover a little bit more than $k$ points.  To deal with this issue, whenever a tour of the DP solution covers more than $k$ points we carefully choose enough points and drop them from the tour to make it feasible. Finally we compute a solution just on the dropped points using a $3$-approximation \cite{hr85} and prove that it has negligible length compared to $\opt$.

In section \ref{sec:main} we present our approximation scheme and prove its correctness under the assumptions that the DP solution is near optimal (Theorem \ref{thm:black}, proved in section \ref{sec:blackcost}) and that the $3$-approximation solution on the dropped points has length at most $O(\epsilon)\opt$ (Theorem \ref{thm:red}, proved in section \ref{sec:redcost}).  Section \ref{section:dynamicprogram} presents the DP and section \ref{sec:derandomize} the derandomization.  

\smallskip
\noindent
{\bf Related Work. } Our work builds on the approach that Arora \cite{arora} used in designing a PTAS for the geometric traveling salesman problem. Similar techniques were also presented by Mitchell \cite{mitchell}. Recently these techniques have been applied to design approximation schemes for several NP-Hard geometric problems, including the polynomial time approximation schemes for Steiner Forest \cite{steinerforest}, and $K$-Median \cite{kmedian} and quasipolynomial time schemes for Minimum Weight Triangulation \cite{mintriangulation} and Minimum Latency problems \cite{minlatency} among others. See \cite{arorasurvey} for a survey of these techniques and a discussion about generalizing them to other problems.
 

\section{The Algorithm}\label{sec:main}
An overview of our approximation scheme is given in Algorithm \ref{alg:ktour}. We use a quasipolynomial time DP to find a near optimal solution, $\opt^{DP}$, which may include some tours that cover more than $k$ points. We drop points from each infeasible tour of $\opt^{DP}$, choosing the points carefully using a randomized procedure to obtain the set of feasible \emph{black} tours. A solution on the dropped points, the {\em red} tours, is obtained using a $3$-approximation. We output the union of the red and black tours. 

\begin{algorithm}[htpb]\caption{CVRP approximation scheme}\label{alg:ktour}
Input: $n$ points in $\mathbbm{R}^2$ and an integer $k$. 
\begin{algorithmic}[1]
\STATE
Perturb instance, perform random dissection and place portals as described in section \ref{sec:preprocess}.
\STATE 
Use the DP from section~\ref{section:dynamicprogram} to find $\opt^{DP}$ which is defined in subsection \ref{sec:structure}.
\STATE 
Trace back in the DP's history to construct tours and assign types to points using the randomized type assignment from subsection \ref{sec:typeassignment}. 
\STATE 
A point is \emph{black} if has type -1 and \emph{red} otherwise. Drop all red points from the $\opt^{DP}$ tours.
\STATE
Get a solution for the red points using the $3$-approximation Algorithm from subsection \ref{sec:3approximation}. 
\end{algorithmic}
Output: the union of the red tours on the red points and the black tours on the black points.
\end{algorithm}
\subsection{Preprocessing~\cite{arora}}\label{sec:preprocess}
\smallskip
\noindent
{\bf Perturbation.} Algorithm \ref{alg:ktour} will work on a perturbed instance. Following Arora's approach from \cite{arora}, we perturb and scale so that all points fit into a square bounding box of side length $L = O(n)$ where $L$ is a power of $2$, and so that all points have integer coordinates and the distance between any two points is either $0$ or at least $4$.  A solution for the perturbed instance can be extended to the original instance using additional length $O(\epsilon)\opt$. See appendix \ref{sec:tools} for the details. From this point, $\opt$ will denote the optimal CVRP solution to the perturbed instance.

\smallskip
\noindent
{\bf Randomized Dissection.} Using Arora's divide and conquer paradigm from ~\cite{arora}, we perform a randomized dissection of the input by recursively partitioning the bounding square into 4 smaller squares of equal size using one horizontal and one vertical dissection line. The recursion stops when the smallest squares have size $1 \times 1$. We define levels for the squares and the dissection lines. See appendix \ref{sec:tools} for  the details. The maximum level of any square or line is $\ell_{\max} = O(\log n)$.   The probability that a line $l$ becomes a level $\ell$ dissection line in the randomized dissection is 
\begin{equation}\label{eq:prline}
Pr(\level(l) = \ell) = 2^{\ell}/L
\end{equation}  

\smallskip
\noindent
{\bf Portals.} As in~\cite{arora}, we place points called \emph{portals} on the boundary of dissection squares that will be the entry and exit points for tours.  Let $m = O(\log n/\epsilon)$ and a power of $2$. Place $2^{\ell}m$ portals equidistant apart on each level $\ell$ dissection line for all $\ell \le \ell_{max}$. Since a level $\ell$ line forms the boundary of $2^\ell$ level $\ell$ squares there will be at most at $4m$ portals along the boundary of any dissection square $b$. As $m$ and $L$ are powers of $2$, portals at lower level squares will also be portals in higher level squares.  

\begin{definition}(Portal respecting and light) 
A tour is {\em portal respecting} if it crosses dissection lines only at portals. A tour is {\em light} if it is portal respecting and crosses each side of a dissection square at most $r= O(1/\epsilon)$ times. 
\end{definition}

Arora showed that for TSP there exists a near optimal solution which is light.  
\begin{theorem}\label{thm:arora}\cite{arora}
Let $\opt(TSP)$ denote the optimal solution for an instance of Euclidean TSP and let $D$ be its randomized dissection. With probability at least $1/2$ there exists a salesman tour of length $(1+O(\epsilon))\opt$ that is light with respect to dissection $D$.
\end{theorem}

\subsection{The Structure Theorem}\label{sec:structure}
We formally define the solution computed by our DP. Recall that remembering the exact number of points on tour segments may require too many bits. Instead, we define $O(\log k)$ \emph{thresholds} that are multiples of $(1+\epsilon/\log n)$ in the range $[1, k]$ and round down the number of points on each segment to its closest threshold. Given a tour segment covering a non-threshold number of points $x$ with $t$ being the closest threshold value less than $x$, to "round" the segment we set the \emph{type} of exactly $x-t$ points to indicate they should  be dropped from the segment.  When two tour segments each covering a threshold number of points are concatenated, the new segment may have to be "rounded" again as the sum of two thresholds do not necessarily add to another threshold.  The DP works bottom-up in the dissection tree and may have to "round" tour segments at each level of the tree. It will mark the points that need to be dropped at level $\ell$ by setting their type to $\ell$. Points with a type in $[0,\ell_{max}]$ are dropped from the final tours.
  
\begin{definition}(Thresholds, Types, and Rounded Segments)
\begin{itemize}
\item
Let $\tau = \log_{(1+\epsilon/\log n)}(k/\epsilon)$. We define a sequence of $\tau +1$ {\em thresholds} :  
$t_0 =1$, $t_1=1/\epsilon$, \ldots, $t_i =1/\epsilon(1+\epsilon/\log n)^{i-1}$, \ldots, $t_{\tau} = k$. 
\item 
Assume that each point has a \emph{type} which is an integer in $[-1, \ell_{max}]$. We say that a point is \emph{active at level $\ell$}  if its type is strictly less than $\ell$.
\item
Let $\pi = (\pi_i)$ be a set of tours.  For any $\pi_i$ and any dissection square $b$ at level $\ell$, a {\em segment} is a connected component of $\pi_i\cap b$. A segment is either {\em rounded} or {\em unrounded}, with the following property: for each rounded segment $s$, there is a threshold $t_i$ such that $s$ covers exactly $t_i$ active points.
\end{itemize}
\end{definition}

The DP builds tours are allowed to cover more than $k$ points and thus in one sense the DP solves a \emph{relaxed} version of CVRP. To ensure that the additional cost of making the DP solution feasible small compared to $\opt$, we only round tour segments inside a dissection square when there are many, at least $\gamma$ (defined below), segments entering the square in which case the cost of going from the depot to the dropped points in the square can be \emph{charged} to $\opt$. If there are only a few, less than $\gamma$ per threshold, tour segments entering the square, we can afford to remember the exact number of points per segment. The intuition behind the third part of definition \ref{def:typerespecting} is to limit the number of points that are marked to be dropped from each tour.
\begin{definition}(Relaxed CVRP) \label{def:typerespecting}
A {\em relaxed CVRP} is a set of tours such that there exists an assignment of types to the points with the following properties:
\begin{enumerate}
\item 
Each tour covers the depot, at most $k$ points of type $-1$, and possibly some points of type $>-1$. The union of the tours covers all $n$ points.
\item 
Let $\gamma= \log^9 n/\epsilon^4$. In any dissection square $b$ for each threshold value $t_i$ ($i \le \tau-1$) there are at most $\gamma$ unrounded segments covering between $[t_i, t_{i+1})$ points. The number of rounded segments in $b$ covering exactly $t_i$ points that are active at $\level (b)$ is an integer multiple of $\gamma$.
\item 
Let $b$ be a dissection square and let $s$ be a tour segment in $b$, which has $t$ points that are active at $\level (b)$. Then segment $s$ has at most $t(1+\epsilon/\log n)$
active points at $\level(b)+1$.
\end{enumerate}
\end{definition}

We apply the concept of being light to each tour of the solution individually rather than to the solution as a whole.
\begin{definition}\label{def:ilight}
 Let $D$ be a dissection and let $S$ be a CVRP solution consisting of tours $(\pi_i)$. $S$ is called {\em i-light} if each tour $\pi_i$ is light. 
\end{definition}

We extend the objective function to include a penalty for tour crossings:
\begin{definition} (Extended Objective Function)
Let $\pi=(\pi_i)$ be a set of tours. For every level $\ell$ let $c(\pi_i,\ell)$ be the number of  times  tour  $\pi_i$ crosses the boundary of level $\ell$ squares, and $d_{\ell} = L/2^{\ell}$ be the length of a level $\ell$ square in the dissection. The extended objective function is:
\begin{equation}\label{eq:F}
F(\pi) = \sum_i\length(\pi_i) + \frac{\epsilon}{\log^2n }\sum_{\textrm{ level}  \ell} \sum_i c(\pi_i, \ell) \cdot d_\ell,
\end{equation} 
\end{definition}

\begin{theorem}\label{thm:black}(Structure Theorem) In expectation over the shifts of the random dissection, the i-light and relaxed CVRP solution which minimizes the extended objective function $F$ given in equation \ref{eq:F} has length $(1+O(\epsilon))\opt$.
\end{theorem}

We prove Theorem \ref{thm:black} in section~\ref{sec:blackcost}. Let $\opt^{DP}$ denote the relaxed and i-light CVRP solution that minimizes the extended objective function $F$ given in equation \ref{eq:F}. We present a DP to compute $\opt^{DP}$ in section \ref{section:dynamicprogram}.

\subsection{A constant factor approximation~\cite{hr85}}\label{sec:3approximation}
For a solution of the red points we use the constant factor approximation of Haimovich and Rinnooy Kan which partitions a traveling salesman tour of the points into tours that cover at most $k$ points \cite{hr85}. Algorithm \ref{alg:3approx} below is a version of the Algorithm presented by \cite{hr85} and it is known to be a $3$-approximation.
 
\begin{algorithm}[htbp]\caption{TSP Partitioning 3-approximation~\cite{hr85}}\label{alg:3approx}
Input: $n$ points and the depot in $\mathbbm{R}^2$ and an integer $k$. 
\begin{algorithmic}[1]
\STATE 
Let $\pi$ denote a tour of input points and the depot obtained using a $2$-approximation of TSP. 
\STATE 
Choose a point $p$ uniformly at random from $\pi$.
\STATE
Go around $\pi$ starting at $p$, and every time $k$ points are visited, take a detour to the depot.
\end{algorithmic}
Output the resulting set of $\lfloor n/k \rfloor + 1$ tours as the CVRP solution.
\end{algorithm}
    
\begin{theorem}\label{thm:3approx} \cite{hr85,aktt}
Let $I$ denote the set of input points, $o$ the depot, and $d(i,o)$ denote the distance of point $i$ from the depot.  Define $\Rad(I) = 2/k \sum_{i \in I}d(i,o)$ and let $TSP(I \cup \{o\})$ denote the length of the minimal tour of $I$ and $o$. We have that: 
\begin{itemize} 
\item $Rad(I) \le \opt$,
\item $TSP(I \cup \{o\}) \le \opt$ and
\item in expectation Algorithm \ref{alg:3approx}'s solution has length $Rad(I) + 2TSP(I \cup \{o\}) \le  3\opt$.
\end{itemize}
\end{theorem}

\subsection{Assigning Types}\label{sec:typeassignment}
The 3-approximation on the red points has small length if the $\Rad$ and $TSP$ of the red points have small value. To ensure this is the case, when selecting $y$ points to drop from segment $S$, we choose $y$ such that $\length(y) \le O(\epsilon)\length(S)$, and such that average distance of points in $y$ to the depot is only a $O(\epsilon)$ fraction of the average distance of points on segment $S$ to the depot. Section \ref{sec:redcost} shows that both conditions hold if $y$ is chosen as follows: 

\begin{algorithm}[htbp]\caption{Randomized Type Assignment Procedure}\label{alg:typeassignment}
Input: A tour segment $S$ from a level $\ell$ square $b$ containing $S_a$ active points and requiring $y$ active points to be dropped
\begin{algorithmic}[1]
\STATE 
Select one active point $p$ uniformly at random from $S_a$
\STATE 
Starting at point $p$, select the next $y-1$ points from $S_a$ that lie consecutively after $p$ on the segment $S$; 
if the last active point on the segment before $b$'s boundary is reached without having selected $y$ points, wrap around and select active points from the other end of the tour segment (after $S$ enters $b$)
\STATE 
Label each of the $y$ chosen points with type $\ell$.
\end{algorithmic}
\end{algorithm}

\subsection{Proof of Theorem \ref{thm:main}}
Lemma~\ref{lem:dp} proves the correctness and running time of a DP that computes $\opt^{DP}$. See section~\ref{section:dynamicprogram} for its proof.  
\begin{lemma}\label{lem:dp}(Dynamic Program) 
Given the set of input points and a randomly shifted dissection, the dynamic program of Section~\ref{section:dynamicprogram} finds an i-light and relaxed CVRP solution that minimizes the objective $F$ defined in Equation~(\ref{eq:F}) in time $n^{\log^{O(1/\epsilon)}n}$. 
\end{lemma} 

The length of the solution output by Algorithm \ref{alg:ktour} will be the sum of the lengths of the black tours and red tours. Since the black tours are obtained by dropping points from the DP solution they have length at most $\opt^{DP}$, which by Theorem \ref{thm:black} is at most $(1 + O(\epsilon))\opt$. Theorem \ref{thm:red}, which is proved in section~\ref{sec:redcost}, shows that the length of the red tours is $O(\epsilon)\opt$.

\begin{theorem}\label{thm:red} In expectation over the random shifts of the dissection and the random type assignment the length of the red tours output by Algorithm \ref{alg:ktour} is $O(\epsilon)\opt$.
\end{theorem}

Thus the solution output by Algorithm \ref{alg:ktour} has total length $(1 + O(\epsilon))\opt$. The DP dominates the running time. The derandomization of the Algorithm is discussed in section \ref{sec:derandomize}.

\section{The Dynamic program}\label{section:dynamicprogram}
This section presents the quasi-polynomial time dynamic program and proves Lemma~\ref{lem:dp}. The dynamic program (DP) finds $\opt^{DP}$, the relaxed and i-light solution that minimizes the extended objective from Equation~(\ref{eq:F}).

\smallskip
\noindent
{\bf The DP table.} A {\em configuration} $C$ of a dissection square $b$ is a list of entries: for each pair of portals $p,q$, the configuration has two sublists of entries, one to record information about rounded tour segments and one to record information about unrounded tour segments:
\begin{enumerate}
\item first sublist: $(r_1^{p,q}, \ldots, r_i^{p,q}, \ldots r_{\tau}^{p,q})$, where $r_i^{p,q}$ is the number of rounded tour segments that use portals $p$ and $q$ and cover exactly $t_i$ active points
\item second sublist: $(u_1^{p,q}, \ldots, u_j^{p,q}, \ldots, u_{\gamma \cdot \tau}^{p,q})$, where $u_{j}^{p,q}$ is the number of active points covered by the $j$-th unrounded tour segment that uses portals $p$ and $q$ 
\end{enumerate}

The DP has a table entry $L_b[C]$ for each dissection square $b$ and each configuration $C$ of $b$. The table cell $L_b[C]$ stores the minimum cost (according to objective $F$ as defined in Equation \ref{eq:F}) of placing tour segments in $b$ in a way which is compatible with $C$ and with the relaxed and i-light Definitions~\ref{def:typerespecting} and \ref{def:ilight}. $\opt^{DP}$ is the minimum table entry over all configurations of the root level square. 

\smallskip
\noindent
{\bf Computing the table entries.} The table entries are computed in bottom-up order, in the following manner. 


Inductively, let $b$ be a square at level $\ell$ and let $b_1, b_2, b_3, b_4$ be the children of $b$ at level $\ell +1$. Since every tour is i-light, a tour segment inside $b$ crosses the boundaries inside $b$ between $b_1, b_2, b_3, b_4$ at most $4r$ times, and always through portals. 
Thus the segment is the concatenation of at most $4r+1$ pieces, where a piece goes from some portal $m_i$ to some portal $m_{i+1}$ in one of the children of $b$. By the relaxed definition applied to level $\ell +1$, each piece is either rounded or one of the $\gamma  \tau$ unrounded tours inside a child of $b$. Thus every piece can be described by a tuple $(p,q,x)$, where $p,q$ are portals and $x$ is either one of the rounded threshold values $t_i$ for some $ i<\tau$ or a number $j \le \tau\gamma$ indicating it is the $j$-th unrounded tour in a child square of $b$. The {\em profile} $\Phi= (p, m_1, n_2), (m_1, m_2, n_1), \ldots (m_v, p', n_v)(f)$ of the segment is the list of those $4r+1$ tuples (representing tour segment pieces), plus a flag $f$ which is true iff the segment is rounded.
Consider a profile $\Phi$ of $b$ with flag $f=1$. Suppose that the concatenation of the pieces described by $\Phi$ contains $x$ active points. Let $t_i$ be the threshold value defined by: $t_i\leq x < t_i(1 + \epsilon/\log n)$. Then the DP counts this segment as having $t_i$ active points. 


Let $D$ denote the number of possible profiles for a segment in square $b$. For each profile $\Phi$, let  $n_{\Phi}$ denote the number of tour segments in $b$ with profile $\Phi$. An \emph{interface vector} $I=( n_{\Phi })_{\Phi }$ is a list of $D$ entries.
 Intuitively the vector $I$, provides the interface between how tour segments in $b$ are formed by concatenating the segments of $b$'s children. 

\smallskip 

Let $C_0$ be a configuration for square $b$. The calculation of $L_b(C_0)$ is done in a brute force manner by iterating through all possible values of the interface vector $I$ and all possible combinations of configurations in $b$'s children, $C_1, C_2, C_3, C_4$.  
A combination $C_0, I, C_1, C_2, C_3, C_4$ is {\em consistent} if $I$ describes at most $\gamma\cdot\tau$ unrounded segment and if gluing $C_1, C_2, C_3, C_4$ according to $I$ yields configuration $C_0$. 


The cost of a consistent combination is computed in the following way: The cost of configurations $C_1, \ldots C_4$ is stored in lookup tables $L_{b_i}(C_i)$, $1\leq i\leq 4$. Let $c_b$ denote the total number of tour segments in $b$ as specified by $I$ (also specified by $C_0$).
The value of objective function $F$, as defined by equation \ref{eq:F} of $(C_1, C_2, C_3, C_4, I)$ is the sum of the costs of $C_i$ for child square $b_i$, plus $(\epsilon/\log^2n)2c_b(I)$.   The value of $L_b(C_0)$ is given by the tuple $(C_1, C_2, C_3, C_4, I)$ consistent with $C_0$ with minimum cost. 

\smallskip
\noindent
{\bf Analysis of the dynamic program.}
How many possible configurations are there for a square $b$? There are $O(\log^2n)$ different pairs of portals ($p,q$); for each $(p,q)$, there are $\tau$ entries $r_{i}^{p,q}$ in the first sublist and $\tau\gamma$ entries $u_{j}^{p,q}$ in the second sublist. Thus a configuration of $b$ is a list of $O(\tau \gamma \log^2n)$ entries. Each entry (the number of segments $r_{i}^{p,q}$ or number of points $u_{j}^{p,q}$) is an integer between $0$ and $n$, thus the total number of configurations for square $b$ is  $n^{O(\tau \gamma \log^2n)} = n^{O(\log^{12} n)}$. As there are $O(n^2)$ dissection squares, the DP table has size  $n^{O(\log^{12}n)}$  overall. 

How many possible profiles are there for a segment in square $b$? $\Phi$ has a list of $O(r)$ tuples $(p,p',x)$. There are $O(\log^2n)$ choices of portals $p,p'$ and $O(\tau\gamma)$ choices of $x$, so there are $O(\tau\gamma\log^2n)$ possibilities for each tuple. The flag doubles the number of profiles so there are $D =(\log^{12}n)^{O(r)}$ possible profiles. As $r = O(1/\epsilon)$, $D=\log^{O(1/\epsilon)}n$. 

How many possible interfaces are there for a square $b$? At most $n^D$, since each $ n_{\Phi }$ is in $[0, n]$. This means we have only a quasi-polynomial number of possibilities for the interface vector $I$ for square $b$.

Checking for consistency takes time polynomial in the size of the list of entries in $I$ and $C_i$. 

There are $n^{\log^{O(1)}n}$ possible values for each $C_i$ and $n^{\log^{O(1/\epsilon)}n}$ possible values for $I$. Thus in total it takes time polynomial in $n^{\log^{O(1/\epsilon)}n}$ to run through all combinations of $I, C_1, C_2, C_3, C_4$ and to compute the lookup table entry at $L_b[C_0]$.

\smallskip
\noindent
{\bf Remark.}
  The DP verifies the existence of a type-assignment satisfying definition \ref{def:typerespecting} but does not actually label points with a specific type. It merely records the number of active points that tour segments it constructs should have. Once the cost of $\opt^{DP}$ solution is found, we can trace through DP solution's history, and find a valid type assignment by looking at the decisions made by the DP. In fact the type assignment can be done during the same time that the tours of $\opt^{DP}$ are constructed.  For example while constructing $\opt^{DP}$, if we build a tour segment with $x$ active points at level $\ell$ but the DP's history recorded the segment as having $t$ active points, we can choose any $x-t$ active points from the segment and label them with type $\ell$. \footnote{Labelling \emph{any} active points on the segment with type $l$ will satisfy the relaxed CVRP definition. But we use the randomized type assignment procedure \ref{alg:typeassignment} to ensure that a solution on these points will have small cost.}

\section{Proof of Structure Theorem}\label{sec:blackcost}
Let $\opt^L$ denote the i-light solution of minimum length. In Lemma \ref{lem:foptdp} we show that $F(\opt^{DP}) \le (1 + O(\epsilon))\opt^L$, where $F$ is defined in equation~\ref{eq:F}. As $\opt^{DP}$ is at most $F(\opt^{DP})$, we get that $\opt^{DP} \le (1 + O(\epsilon))\opt^L$.  Then we apply Corollary \ref{cor:aroraExtention}, given below, to show that $\opt^{DP} \le (1+O(\epsilon)\opt)$ and prove our structure Theorem. Corollary \ref{cor:aroraExtention} is a simple generalization of Arora's structure Theorem that shows that $\opt^L$ is near optimal. 

\begin{corollary}(Generalization of Arora)\label{cor:aroraExtention} In expectation over the random shifts of the dissection,
$E[\opt^{L}] \le (1 + O(\epsilon))\opt$
\end{corollary}
\begin{proof}
Let $\opt^L$ consist of a set of tour $\pi = \pi_1, \ldots \pi_m$. Apply Arora's structure Theorem \ref{thm:arora} to each tour, sum, and use linearity of expectation. 
\end{proof}

\begin{lemma}\label{lem:foptdp}
 In expectation over the random shifted dissection, $F(\opt^{DP}) \le (1 + O(\epsilon))\opt^L$.
\end{lemma}
\begin{proof}
To start comparing $\opt^{DP}$ and $\opt^{L}$, we apply Lemma \ref{lem:assigntypes} to turn $\opt^{L}$ into a solution that satisfies the relaxed definition \ref{def:typerespecting}. See appendix \ref{sec:relaxed} for proof of Lemma \ref{lem:assigntypes}

\begin{lemma}\label{lem:assigntypes}
Let $S$ be a CVRP solution on input $I$. There exists an assignment of types to points, that turns $S$ into a solution that satisfies definition \ref{def:typerespecting}. The tours of $S$ are not modified and the length of $S$ remains unchanged.
\end{lemma}  

Now $\opt^{DP}$ and $\opt^L$ are both i-light and relaxed solutions. As $\opt^{DP}$ also minimizes objective function $F$, we have
$$F(\opt^{DP}) \le F(\opt^{L}) =  \opt^{L} + (\epsilon/\log^2n)\sum_{ \textrm{ level } \ell} c(\pi^{L}, \ell) d_{\ell}$$  
where $\pi^L$ are the tours of $\opt^L$. Now we only need to show that the last term summing the number of crossings in $\opt^L$ is $O(\epsilon)\opt^{L}$ in expectation. Lemma \ref{lem:crossings} allows us to bound the number of crossings in $\opt^L$ in terms of the number of crossings in $\opt$, and Lemma \ref{lem:lbopt} allows us to charge each crossing to the length of $\opt$. See appendix \ref{sec:extentionsArora} for the proofs.

\begin{lemma}\label{lem:lbopt} In expectation over the random dissection, for any level $\ell$
$\opt \ge O(d_{\ell}) E[(c(\pi, \ell)]$.
\end{lemma}

\begin{lemma} \label{lem:crossings} For a random dissection at any level $\ell$, $E[ c(\pi^{L}, \ell)] \le O(\log n) E[ c(\pi, \ell)]$.
\end{lemma}

Applying Lemma \ref{lem:crossings} we get,  
\begin{equation}\label{eq:costdp1}
(\epsilon/\log^2n)\sum_{ \textrm{ level } \ell} E[c(\pi^{L}, \ell)] \cdot d_{\ell} \le (\epsilon/\log^2n)\sum_{ \textrm{ level } \ell} O(\log n \cdot d_{\ell})E[ c(\pi, \ell)] 
\end{equation}

By Lemma \ref{lem:lbopt}, $E[c(\pi, \ell)] O(d_{\ell})$ is at most $\opt$ and as $\ell_{max}=O(\log n)$, Equation \ref{eq:costdp1} is $O(\epsilon)\opt$. This proves the Lemma as $\opt \le \opt^{L}$. 
\end{proof}

\section{Proof of Theorem~\ref{thm:red}}\label{sec:redcost}
Let $R$ denote set of the points  marked red by Algorithm \ref{alg:ktour}. By Theorem \ref{thm:3approx}, the $3$-approximation on $R$ has cost at most $\Rad(R) +\TSP(R \cup \{o\})$, where $o$ is the depot. Lemmas \ref{lem:rad} and \ref{lem:tsp} proves Theorem \ref{thm:red} by showing that in expectation both quantities are $O(\epsilon)\opt$.  

\begin{lemma} \label{lem:rad} In expectation over the random type assignment, $\Rad(R) = O(\epsilon)\opt$ \end{lemma} 
\begin{lemma}\label{lem:tsp} In expectation over the random dissection and type assignment $\TSP(R \cup \{o\})= O(\epsilon)\opt$ \end{lemma}

\subsection{Properties of the randomized type assignment procedure}
We state some properties that will be useful in proving the Lemmas \ref{lem:rad},\ref{lem:tsp}. See appendix \ref{sec:randomtypeassignment} for the proofs.

Let $b$ be a level $\ell$ square containing points labelled type $\ell$ by Algorithm \ref{alg:ktour} and $S$ be a a rounded tour segment inside $b$. Let $R_s = r_1, r_2, \ldots r_d$ be the interval of points labelled type $\ell$ on $S$ ($|R_s|$ may be zero) and $S_a = s_1, s_2, \ldots s_x$ be the active points on $S$ prior to rounding segment $S$. We now list some Properties of the rounded segment $S$.  

\begin{property}\label{prop:epsilonnumber} $|R_s| \le |S_a|\cdot O(\epsilon/\log n)$ \end{property}

\begin{property}\label{prop:probofred}
A point $s\in S_a$ is in $R_s$ with probability $|R_s|/|S_a|$.
\end{property}

\begin{property}\label{prop:epsilonlength}$E[\length(R_s)] \le \length(S_a) \cdot O({\epsilon}/{\log n})$
\end{property}

\subsection{Proof of Lemma \ref{lem:rad}}
\begin{proof}Recall that,  $\Rad(R) = 2/k \sum_{x \in R} d(x,o)$, where $d(x,o)$ is the distance of point $x$ from the depot.  By Theorem \ref{thm:3approx} $\Rad(I) \le \opt$, so it is sufficient to show that $\Rad(R) \le O(\epsilon) \Rad(I)$.  Fix any level $\ell$ of the dissection and let $R_{\ell}$ be the set of points which were assigned type $\ell$. We show that in expectation $\Rad(R_{\ell}) \le O(\epsilon/\log n) \Rad(I)$. The claim follows by linearity of expectation (over all levels) since $\Rad(R) = \sum_{\textrm{ level } l} \Rad(R_{\ell})$.

\smallskip

Partition $R_{\ell}$ according to the tour segment it is from: $R_{\ell}^1 \subset S_1, R_{\ell}^2 \subset S_2, \ldots R_{\ell}^m \subset S_m$ where $R_{\ell}^j$ is the set of red points from tour segment $S_j$.  By definition we have that
\begin{equation}\label{radI}
\Rad(I) \ge \frac{2}{k} \sum_{j=1}^m \sum_{x\in S_j} d(o,x)
\end{equation} 
As $R_{\ell}$ is picked randomly, and by Properties \ref{prop:probofred} and \ref{prop:epsilonnumber}, $\Pr[x\in R_{\ell}^j] \le O(\epsilon/\log n)$, so we get
$$E[\Rad(R_{\ell})] = \frac{2}{k} \sum_{j=1}^m \sum_{x\in S_j} d(o,x) \Pr[x\in R_{\ell}^j] \le  \frac{2}{k}\sum_{j}^m \sum_{x\in S_j} d(o,x) \cdot O(\epsilon/\log n)$$ 
Combining this with equation \ref{radI} we get that $E[\Rad(R_{\ell})] \le  O(\epsilon/\log n)\Rad(I)$.
\end{proof}

\subsection{Proof of Lemma \ref{lem:tsp}}
\begin{proof}
Let $R_{\ell}$ be the points labeled type $\ell$ at level $\ell$. We show that $E[\TSP (R_{\ell} \cup \{ o\})] \le O(\epsilon/\log n) \opt$.
This implies Lemma \ref{lem:tsp} since the tours of $\{R_{\ell} \cup \{o\}\}$ from all levels can be pasted together at the depot to yield a tour of $(R\cup \{o\})$. 

\smallskip

Let $B_{\ell}$ be the squares at level $\ell$ containing points of $R_{\ell}$. We consider the cost of $\TSP(R_{\ell} \cup \{o\})$ in two parts: the outside and inside costs, where intuitively the outside cost will be the cost to get to the squares $B_{\ell}$ from the depot and the inside cost will be the cost of visiting the red points inside the square.  focusing on the outside cost, let $C$ be a set of points containing at least one portal from each square of $B_{\ell}$ such that $MST(C \cup \{o\})$ is minimized\footnote{ $C$ is used only for the analysis and does not need to be found explicitly}. The optimal tour of $R_{\ell} \cup \{O\}$ is at most $2MST(C \cup \{o\})$, thus

\begin{equation}\label{eq:tspRl} 
TSP(R_{\ell} \cup \{o\}) \le 2MST(C \cup \{o\}) + \sum_{b\in B_{\ell}} \textrm{ inside cost of } b 
\end{equation}

Claims \ref{claim:outsidecost} and \ref{claim:insidecost} prove that the quantities on the right hand side of equation \ref{eq:tspRl} are both $O(\epsilon/\log n)\opt$, proving that $TSP(R_{\ell} \cup \{o\}) \le O(\epsilon/\log n)$.
\end{proof}  

\begin{claim}\label{claim:outsidecost} In expectation over the random shifted dissection,
$E[2MST(C \cup \{o\})] \le O(\epsilon/\log n)\opt$.
\end{claim}
\begin{proof}
We will show that, $MST(C \cup \{o\}) \le O(\epsilon/\log n)\opt^{DP}$. The claim follows, as $\opt^{DP} \le (1 + O(\epsilon)\opt$, by the structure Theorem \ref{thm:black}. 

\smallskip

Consider the fully connected graph $G$ with one vertex for each point in $C$ and one more for the depot. Define the weight of an edge of $G$ to be the distance between the two vertices connected by that edge. Consider the following linear program, $A$ with value $v$ on $G$. 
$$v=\min_x (w_e\cdot x_e) ~~\hbox{ s.t. }
\left\{ \begin{array}{lll}
\sum_{e\in \delta(S)} x_e & \geq \gamma/4r & \forall S \subset V\\
x_e& \geq 0 &
\end{array} \right.
$$
As each $b\in B_{\ell}$ contains some points labelled $\ell$ (i.e at least a group of rounded segments), $\opt^{DP}$ contains are least $\gamma$ tour segments crossing into $b$. Since each tour in $\opt^{DP}$ is i-light, there are at least $\gamma/4r$ tours entering $b$. Thus $\opt^{DP}$ has at least $\gamma/4r$ edges crossing any cut separating the depot from a point in $C$. As $v$ is the minimum cost way to have at least $\gamma/4r$ edges cross all such cuts, $\opt^{DP} \ge v$. 

\smallskip

Now consider the linear program $A'$ below. $A'$ is the relaxation of the IP for MST. Let $v'$ be the value of $A'$ on graph $G$.
$$v' = \min_x (w_e\cdot x_e) ~~\hbox{ s.t. }
\left\{ \begin{array}{lll}
\sum_{e\in \delta(S)} x_e & \geq 1 & \forall S \subset V\\
x_e& \geq 0 &
\end{array} \right.
$$

Observe that for any solution $x$ of $A$, $x'= x \cdot 4r/\gamma$ is a solution for $A'$. As $A$ and $A'$ have the same objective, $v \cdot 4r/\gamma = v'$. The MST relaxation $A'$ is known to have integrality gap at most $2$ \cite{vaz}, so that $v' \ge 2\textrm{ MST } (C\cup \{o\})$. Thus we have that
$$\opt^{DP} \ge v = v'\cdot (\gamma/4r) \ge  MST(C \cup \{o\})\cdot (\gamma/2r)$$ 
Thus $(2r/\gamma)\cdot\opt^{DP}\ge MST(C\cup \{o\})$ and as $2r/\gamma = o(\epsilon/\log n)$, the Claim is proved.
\end{proof}

\begin{claim}\label{claim:insidecost} In expectation over the random dissection and the random type assignment the total inside cost at level $\ell$ is at most $O(\epsilon/\log n)\opt$.
\end{claim}
\begin{proof}
The inside cost at level $\ell$ is the sum of the inside costs of each square $b \in B_{\ell}$. The contribution of square $b \in B_{\ell}$ to the inside cost, is the sum the length of the intervals of type $\ell$ points inside $b$ plus the cost of connecting these intervals to the boundary of $b$. 
In Claim \ref{claim:redlengths} we show that, in expectation over the random type assignment, the sum over all squares in $B_{\ell}$ of the length of intervals of type $\ell$ points is $O(\epsilon/\log n)\opt^{DP}$.  The type $\ell$ intervals inside each $b \in B_{\ell}$ must be connected to each other and to the boundary of their square. We refer to $CC(\ell)$ as the total the connection cost at level $\ell$. $CC(\ell)$ is the sum of the length of the boundaries of each square $b\in B_{\ell}$ plus the cost of connecting the type $\ell$ intervals inside each $b\in B_{\ell}$ to the boundary of $b$.  Claim \ref{claim:connectioncost} shows that $CC(\ell) = O(\epsilon/\log n)F(\opt^{DP})$. By Lemma \ref{lem:foptdp} and Corollary \ref{cor:aroraExtention} $F(\opt^{DP}) \le (1 + O(\epsilon))\opt$ in expectation, proving this Claim.
\end{proof}

\begin{claim}\label{claim:redlengths} In expectation over the random type assignment the length of all intervals of type $\ell$  is  $O(\epsilon/\log n)\opt^{DP}$. 
\end{claim}
\begin{proof}
Sum the lengths of intervals of type $\ell$ over all squares in  $B_{\ell}$ squares, use Property \ref{prop:epsilonlength} and linearity.
See appendix \ref{sec:randomtypeassignment} for a detailed proof.
\end{proof}

\begin{claim}\label{claim:connectioncost}
The total connection cost for level $\ell$ is $O(\epsilon/\log n)F(\opt^{DP})$. 
\end{claim}
\begin{proof}
Let $CC(\ell)$ denote the total connection cost which is the cost of  connecting the type $\ell$ intervals inside each  $b \in B_{\ell}$ to the boundary of the square. Focus on one square $b \in B_{\ell}$. As $\opt^{DP}$ is relaxed, $b$ has $g_b\cdot \gamma$ rounded segments for some integer $g_b > 0$. Consider a group of $\gamma$ rounded segments. Let $R'$ be a set containing one type $\ell$ point from each of the $\gamma$ segments in the group. The cost to connect all red intervals in the group with the boundary of the square is at most $MST(R') + 5\cdot d_{\ell}$, where $d_{\ell}$ is the side length of square $b$. We bound this using the following bound for TSP~\cite{hr85}\cite{aktt}. (See \cite{hr85} for a proof).

\begin{theorem}\label{thm:tsp}
Let $U$ be a set of points in $2d$-Euclidean space. Let $d_{max}$ be the max distance between any two points of $U$. Then $TSP(U) = O(d_{max} \sqrt{|U|})$
\end{theorem}

In our context, $d_{max} = d_{\ell}$, and $U = R'$. Since $|R'| = \gamma$, $|U| = \gamma$. By Theorem \ref{thm:tsp} we have that $MST(R') + 5\cdot d_{\ell} = O(d_{\ell} \cdot \sqrt{\gamma})$. This holds for each of the $g_b$ groups of rounded segments inside $b$ thus we have that the total connection cost for $b$ is $O(g_b\cdot d_{\ell} \cdot \sqrt{\gamma})$.

\smallskip

Let $g_{\ell}= \sum_{b\in B_{\ell}} g_b$, be the total number of groups of rounded segments at level $\ell$. The total connection cost for level $\ell$ is the sum of the connection costs over all squares $b\in B_{\ell}$,
\begin{equation}\label{eq:totalconnection} CC(\ell) =   O( g_{\ell} \cdot d_{\ell} \cdot \sqrt{\gamma}) 
\end{equation}

Each of the  $g_{\ell} \cdot \gamma$ rounded tour segments intersects twice with the boundary of a level $\ell$ dissection square, thus: $c(\pi^{DP}, \ell) \ge 2g_{\ell}\gamma$, with $\pi^{DP}$ being the tours of $\opt^{DP}$. Using equation \ref{eq:totalconnection}, 
\begin{equation}\label{eq:ccl}
O\left({d_{\ell}}/{\sqrt{\gamma}}\right) \cdot c(\pi^{DP}, \ell)\ge CC(\ell)
\end{equation}
For objective function $F$, (defined in \ref{eq:F}), we have $(\log^2 n/\epsilon) \cdot F(\opt^{DP}) \ge c(\pi^{DP}, \ell)d_{\ell}$.  Substituting it for $c(\pi^{DP}, \ell)d_{\ell}$ in equation \ref{eq:ccl} we get, $O(1/\sqrt{\gamma})(\log^2 n/\epsilon) \cdot  F(\opt^{DP}) \ge CC(\ell) $. As $\gamma = \log^9n/\epsilon^4$ this reduces to $O({\epsilon}/{\log n}) \cdot F(\opt^{DP}) \ge CC(\ell)$, which proves this Claim.
\end{proof}
\section{Derandomization}\label{sec:derandomize}
Arora's dissection can be derandomized by trying all choices for the shifts $a$ and $b$. More efficient derandomizations are given in Czumaj and Lingas and in Rao and Smith \cite{cz98, rs98}.  As for the randomized type assignment Algorithm \ref{alg:typeassignment}, to guarantee that the cost of the dropped points is small, when selecting an interval $Y$ to drop from a segment $S$,  we only need to pick $Y$ such that (1) $\Rad(Y) \le O(\epsilon/\log n) \Rad(S)$  and  (2)$\length(Y) \le O(\epsilon/\log n) \length (S)$. In Lemma \ref{lem:rad} and Property \ref{prop:epsilonlength} we prove that these two conditions hold at the same time, in expectation when $Y$ is chosen by first selecting a point uniformly from $S$ and then selecting the next $|Y|-1$ consecutive points. To derandomize we can test the at most $|S|$ intervals of length $|Y|$ in $S$, (each starting from a different point in $S$), and select any interval that satisfies these two conditions.

\newpage

\appendix

\section{Technical Tools}\label{sec:tools}

\paragraph{The Perturbation} Define a bounding square as the smallest square whose side length $L$ is a power of $2$ that contains all input points and the depot. Let $d$ denote the maximum distance between any two input points.  Place a grid of granularity $d \epsilon/n$ inside the bounding square. Move every input point to the center of the grid square it lies in. Several points may map to the same grid square center and we will treat these as multiple points which are located at the same location. Finally scale distances by $4n/(\epsilon d)$ so that all coordinates become integral and the minimum non-zero distance is least $4$.   

A solution for the perturbed instance can be extended to a solution for the original instance by taking detours from the grid centers to the locations of the points. The cost of such detours will be at most $n \cdot \sqrt{2}d\epsilon/n$. As the two farthest points must be visited from the depot we have that $2d \le \opt$. Thus the total cost of the detours is $\le \epsilon\opt$ and is negligible compared to $\opt$. Note also that scaling does not change the structure of the optimal solution. After scaling the maximum distance between points is $4n/\epsilon$ which is $O(n)$ for constant $\epsilon$.

\paragraph{Randomized Dissection} A dissection of the bounding square is obtained by recursively partitioning a square into 4 smaller squares of equal size using one horizontal and one vertical dissection line. The recursion stops when the smallest squares have size $1 \times 1$. The bounding square has level $0$, the 4 squares created by the first dissection have level $1$, and since $L =O(n)$ the level of the $1\times 1$ squares will be $\ell_{\max} = O(\log n)$.  The horizontal and vertical dissection lines are also assigned levels. The boundary of the bounding square has level $0$, the $2^{i-1}$ horizontal and $2^{i-1}$ vertical lines that form level $i$ squares by partitioning the level $i-1$ squares  are each assigned level $i$. 
A randomized dissection of the bounding square is obtained by randomly choosing integers $a,b \in [0, L)$, and shifting the $x$ coordinates of all horizontal dissection lines by $a$ and all vertical dissection lines by $b$ and reducing modulo $L$. For example the level $0$ horizontal line is moved from $L/2$ to $a + L/2 \mod L$ and the level $0$ vertical line is moved to $b+ L/2 \mod L$. The dissection is "wrapped around" and wrapped around squares are treated as one region. The probability that a line $l$ becomes a level $\ell$ dissection line in the randomized dissection is 
\begin{equation}
Pr(\level(l) = \ell) = 2^{\ell}/L
\end{equation}

\section{Relaxed CVRP} \label{sec:relaxed}

\begin{proof}(Proof of Lemma \ref{lem:assigntypes}) We give an type assignment procedure which initially assigns all points to type $-1$ and never sets any point to a type below $-1$. As $S$ is a valid CVRP solution all tours in $S$ contain at most $k$ points of type $-1$ satisfying the first condition of definition \ref{def:typerespecting}.   The procedure works in a bottom up fashion in the dissection tree from level $\ell_{\max}$ to level $0$.  At the current level $\ell$ consider each dissection square one at a time. For any threshold value $t_i$ for $i \le \tau$ while square $b$ has at least $\gamma$ unrounded tour segments, select exactly $\gamma$ such segments and perform a \emph{group-rounding} as follows: Examine each of the $\gamma$ segments one at a time. If the unrounded tour segment has $x$ active points with $t_i < x < t_{i+1}$, pick any $x-t_i$ of these active points and label them as type $\ell$.
Perform as many group-rounding steps as necessary until square $b$ has at most $\gamma*\tau$ unrounded tours. Proceed similarly to the other squares at level $\ell$.

\smallskip 
The type assignment procedure does not change the construction of any of the tour in $S$ thus the cost of $S$ is unchanged. Now we show that definition \ref{def:typerespecting} is satisfied.  While working at a level $\ell$, the procedure performs group-rounding on each square at level $\ell$ until the square has at most $\gamma*\tau$ unrounded segments. As the group-rounding rounds exactly $\gamma$ segments together there will always be an integer multiple of $\gamma$ rounded segments in each square. Condition (ii) continues to hold while the procedure works on levels $j< \ell$ as in those levels the procedure only labels points with type $j<\ell$ so the number of active points at level $\ell$ remains the same. As for the third condition of definition \ref{def:typerespecting}, note that before rounding a segment at level $\ell$, all points on the segment have type $-1$ or a type greater than $\ell$. Thus prior to rounding the segment will have $x$ active points at level $\ell$ and at level $\ell +1$. Let $t_i \le x \le t_{i+1}$. To round the segment at $\ell$, label $x-t_i$ points with type $\ell$. This still leaves $t_i$ active points at level $\ell$ and $x$ active points at level $l+1$. As $t_i(1 + \epsilon/\log n) = t_{i+1} > x$, the third condition of definition \ref{def:typerespecting} is satisfied. 
\end{proof}

\section{Extensions of Arora's Results to CVRP}\label{sec:extentionsArora}
Let  $t(\pi_j, l)$ denote the number of times a tour $\pi_j$ crosses dissection line $l$. Arora proved that the $\length(\pi_j) \ge \frac{1}{2}\sum_{\textrm {line l}} t(\pi_j, l)$ \cite{arora}. For CVRP as $\sum_{j} \pi_j = \opt$, this implies
\begin{equation}\label{eq:aroraintersection}
\opt \ge 1/2 \sum_{ \textrm{ line } l} t(\pi, l)
\end{equation}

With probability $2^{\ell+1}/L$ a dissection line $l$ forms the boundary for some level $\ell$ square \footnote{The boundaries of level $\ell$ squares are formed by lines at levels $\le \ell$}.  For any level $\ell$ as $E(c(\pi, \ell)) = \sum_{\textrm{ line } l} t(\pi, l) \cdot \Pr[l \textrm{ forms a boundary a level } \ell \textrm{ square }]$, we have
\begin{equation}\label{eq:squarecrossings} 
E(c(\pi, \ell)) = \frac{2^{\ell+2}}{L}  \sum_{\textrm{ line } l} t(\pi, l)
\end{equation}
\begin{proof}(Proof of Lemma \ref{lem:lbopt}) Combine equations \ref{eq:aroraintersection} and \ref{eq:squarecrossings} to get $E(c(\pi, \ell)) \le \frac{2^{\ell+3}}{L}\opt$. The Lemma follows from the fact that a level $\ell$ square has side length $d_{\ell} = L/2^{\ell}$. 
\end{proof}

\smallskip
Let $\pi$ denote the tours of $\opt$ and $\pi^{L}$ the tours of $\opt^L$. Arora gives a procedure to modify $\pi$ into the light tours $\pi^L$ \cite{arora}. But the procedure may create new crossings in $\pi^L$ with dissection squares not present in $\pi$. The next claim bounds the number of these new crossings.
\begin{proof}(Proof of Lemma \ref{lem:crossings}) To modify $\pi$ into $\pi^{L}$ Arora's procedure does bottom up patching to ensure that each tour crosses the edges of dissection squares at most $O(r)$ times.  The second step is to take detours (along the sides of squares) to make the tours portal respecting. Both steps, patching and detouring, can add new crossings to $\pi^{L}$ which are not present in $\pi$.  A patching on edge $e$ of square $b$ adds at most $6$ new crossings to each child square inside $b$ with edge $e$ as a boundary and each detour adds at at most $2$ new crossings to each such child square.  

As Arora's procedure works bottom up the patching and detours taken at all levels $j>\ell$ can add crossings in $\pi^{L}$ with level $\ell$ squares. As there are $2^{\ell-j}$ level $\ell$ squares along the edge of a level $j$ square,  a patching step on the edge of a level $j$ square adds at most $6 \cdot 2^{\ell-j}$ new crossings at level $\ell$ and each detour along these edges add at most $2\cdot 2^{\ell-j-1}$ new crossings at level $\ell$. Thus we have that, 
$$E[c(\pi^{L}, \ell)] \le 6\sum_{j \le \ell} 2^{\ell-j} ( \# \textrm{ patching at level } j) + 2\sum_{j\le \ell} 2^{\ell-j} (\#\textrm{ detours at level } j)$$
A patching step is performed only when a group of at least $O(r)$ crossings are identified in $\pi$, thus an upper bound on the number of expected patching at level $j$ is $E[c(\pi,j)]/O(r)$. Since each crossing at level $j$ can require a detour, we get,
$$E[c(\pi^{L}, \ell)] \le 6\sum_{j \le \ell} 2^{\ell-j} \frac{E[c(\pi, j)]}{O(r)} + 2\sum_{j\le \ell} 2^{\ell-j} E[c(\pi, j)]$$
 
Writing out the definitions of $E[c(\pi, j)]$ using equation \ref{eq:squarecrossings} we see that for all $j\le \ell$, $ E[c(\pi, j)] = E[c(\pi, \ell)] / 2^{\ell-j}$. Substituting this into the above equation simplifies it to, 
$$E[c(\pi^{L}, \ell)] \le 6\sum_{j \le \ell} \frac{E[c(\pi, \ell)]}{O(r)} + \sum_{j\le \ell} E[c(\pi, \ell)]$$

As $\ell \le O(\log n)$, we have that $E[c(\pi^{L}, \ell)] \le O(\log n)E[c(\pi, \ell)](1 + O(1/r))$. This proves the Lemma as $O(1/r) = O(\epsilon)$ is a constant.
\end{proof}
   
\section{Properties of the randomized type assignment procedure}\label{sec:randomtypeassignment}

Let $b$ be a level $\ell$ square containing points labelled type $\ell$ by Algorithm \ref{alg:ktour} and $S$ be a a rounded tour segment inside $b$. Let $R_s = r_1, r_2, \ldots r_d$ be the interval of points labelled type $\ell$ on $S$ ($|R_s|$ may be zero) and $S_a = s_1, s_2, \ldots s_x$ be the active points on $S$ prior to rounding segment $S$. We now list some Properties of the rounded segment $S$.  

\begin{proof}(Proof of Property \ref{prop:epsilonnumber})  In the DP's history, $S$ has a profile $\Phi$ with a flag set to true as $S$ is a rounded segment. Suppose that after $S$ is concatenated according to $\Phi$ it has $x$ active points. Then the DP counts $S$ as a rounded segment having exactly $t_i$ active points for the unique threshold value $t_i$ lying in the interval $[x/(1+\epsilon/\log n), x]$. To get exactly $t_i$ active points on $S$ we would need to set at most $x - t_i \le x(\epsilon/\log n)$ active points to $\ell$. Thus $|R_s|\le x\epsilon/\log n$ while $|S_a| = x$.
\end{proof}

\begin{proof}(Proof of Property \ref{prop:probofred})
Each point  $s\in S_a$  belongs to $|R_s|$ intervals as each interval consists of $|R_s|$ consecutive points. There are a total of $|S_a|$ different intervals, each starting at a different point in $S_a$ and Algorithm \ref{alg:typeassignment} picks uniformly among them.
\end{proof}

\begin{definition} (Length of interval) Let $b_1, b_2$ be the points on the boundary of $b$ where $S$ enters and exits $b$. Then $S$ visits $s_1$ after entering at $b_1$ and it visits $s_x$ before exiting from $b_2$  Let $d(u,v)$ denote the distance between points $u$ and $v$. If $R_s$ does not contain both $s_1$ and $s_x$ then $\length(R_s) = \sum_{i}^d d(s_i, s_{i+1})$.  Otherwise let $r_e=s_x$, then $r_{e+1}=s_1$ as Algorithm \ref{alg:typeassignment} wraps around and $\length(R_s)=\sum_{i= 1}^{e-1} d(s_i, s_{i+1}) + d(s_x, b_2) + d(b_1, s_1) + \sum_{i=e+1}^{d} d(r_{i}, r_{i+1})$.
\end{definition}

\begin{proof}(Proof of Property~\ref{prop:epsilonlength})
Let $b_1$ and $b_2$ be the points where $S$ enters and exits square $b$. Define $z_x = d(b_1, s_1) + d(s_x, b_2)$ and $z_i = d(s_i, s_{i+1})$, for $ 1 \le i < x$. Then the length of $S_a$ is $\sum_{i=1}^{x} z_i$ and 
$$E[\length(R_a)]= \sum_{i=1}^{x} z_i \Pr[z_i \textrm{ is counted inside } R_s]$$
For all $i$ the probability that $z_i$ is counted is the probability that $s_i$ and its consecutive point, $s_{i+1}$,  are both included in $R_s$. (The consecutive point of $s_x$ is $s_1$). A point $s \in S_a$ belongs to exactly $|R_s|$ intervals and the consecutive point of $s$ appears in $|R_s|-1$ of these intervals. Thus the $\Pr[z_i \textrm{ is counted }] = (|R_s|-1)/(|S_a|)$. Applying Property \ref{prop:epsilonnumber},
$$E[\length(R_s)]  =  \sum_{i=1}^{x} z_i \frac{|R_s|-1}{|S_a|} = O\left(\frac{\epsilon}{\log n}\right)\sum_{i=1}^x z_i = O\left(\frac{\epsilon}{\log n}\right)\length(S_a)$$
\end{proof}

\begin{proof}(Proof of Claim~\ref{claim:redlengths})
Sum the lengths of intervals of type $\ell$ over all squares in  $B_{\ell}$ squares, use Property \ref{prop:epsilonlength} and linearity.
Consider a square $b \in B_{\ell}$ and the type $\ell$ points inside $b$, $R_b$ . Partition the points in $R_b$ according to the segments of $b$ they come from: $r_1 \subset s_1, r_2 \subset s_2, \ldots r_m \subset s_m$ such that $r_j$ is the set of points labelled type $\ell$ from tour segment $s_j$. By Property \ref{prop:epsilonlength}, in expectation over the random type-assignment the length of $r_j$ is at most $O(\epsilon/\log n)$ times the length of $s_j$. By linearity, $\sum_{j}^m E[\length(r_j)] \le O(\epsilon/\log n)\sum_{j=1}^m \length(s_j)$.
Let $\opt_b^{DP}$ denote the projection of $\opt^{DP}$ inside square $b$. $\opt^{DP}_b \ge \sum_{j=1}^m {\length(s_j)}$, and  $\opt^{DP}$ is at least the sum of  $\opt_b^{DP}$ over all squares $b \in B_{\ell}$. This implies that the total length of red intervals at level $\ell$ is at most $O(\epsilon/\log n)\opt^{DP}$.  
\end{proof}

\end{document}